\documentclass[a4paper,12pt]{article}
\usepackage{times, url, enumerate}
\usepackage{amsmath,amssymb,graphicx,amsthm,amsfonts,amstext,wasysym,textcomp,xcolor,mathrsfs}
\usepackage[mathscr]{euscript}
\newcommand{\tab}[1][1.5cm]{\hspace*{#1}}
\newtheorem{theorem}{Theorem}[section]

\newtheorem{note}{Note}
\newtheorem{case}{Case}
\begin{document}
	
	\begin{center}
		{\bf\LARGE Enhancing Data Security through Rainbow Antimagic Graph Coloring for Secret-Share Distribution and Reconstruction}
		\vspace{8mm}
		
		{\large \bf Ra\'{u}l M. Falc\'{o}n $^{1}$, K. Abirami $^{2}$, N. Mohanapriya $^{2}$, Dafik $^{3}$ }\\
		\vspace{3mm}

        $^{1}$ Department Applied Mathematics I, School of Architecture, Universidad de Sevilla, 41012 Sevilla, Spain.\\
		$^{2}$ Department of Mathematics, Kongunadu Arts and Science College, Coimbatore-641 029, Tamil Nadu, India.\\
		$^{3}$PUI-PT Combinatorics and Graph, CGANT - University of Jember, Jember, Indonesia.\\
  
		e-mail: \url{abirami7.abi@gmail.com, n.mohanamaths@gmail.com} 
		\vspace{2mm}
	\end{center}

	\textbf{Abstract.} Now-a-days, ensuring data security has become an increasingly formidable challenge in safeguarding individuals' sensitive information. Secret-sharing scheme has evolved as a most successful cryptographic technique that allows a secret to be divided or distributed among a group of participants in such a way that only a subset of those participants can reconstruct the original secret. This provides a safe level of security and redundancy, ensuring that no single individual possesses the complete secret. The implementation of Rainbow Antimagic coloring within these schemes not only safeguards the data but also ensures an advanced level of information security among multi-participant groups. Additionally, the retrieved data is reconstructed and can be disseminated to all group participants via multiple rounds of communication.

\section{Introduction}

\tab Graph coloring problem has evolved as one of the most fundamental graph solving techniques in the field of graph theory and computer networks. Usually graph coloring problems involve coloring the vertices of a graph that are connected by edges, such that no two vertices that are said to be adjacent, receives a same color. For the problem in question, we color the edges of the graph, such that no two edges incident on a common vertex shares the same color. This intricate task is addressed using the renowned \textit{\textquoteleft rainbow coloring\textquoteright} technique for edges, accompanied by an \textit{\textquoteleft antimagic labeling\textquoteright} for the vertices. The amalgamation of these approaches yields the Rainbow Antimagic Connection Number (RACN) for the given graph.\\
\tab The application of the aforementioned coloring technique within a secret-share scheme, facilitates the secure transmission of information or data, such as Passwords, Financial details, Bank account details, etc., through its dispersion of data into $m$ pieces of secret codes. Secret-share schemes enhances various scenarios such as, cryptographic key management, data recovery, access control, blockchain and cryptocurrencies. In a $(k,m)$-threshold secret sharing scheme, $k$ denotes the threshold number of shares required and $m$ denotes the minimum number of participants required to reconstruct the secret. This threshold is set when the secret sharing scheme is established and determines how many participants need to cooperate to reveal the secret. In this secret-share scheme, the aggregation of any $k$ pieces or more guarantees reconstructability, yet the aggregation of even $k-1$ pieces fails to reconstruct the code, thus establishing a robust and secure interconnection among nodes (i.e.,participants). During the reconstruction phase, the specified threshold number of participants must come together, pool their shares, and perform a reconstruction process using the cryptographic methods used during the share generation. Leveraging the principles of RACN, we introduce a method for the secure transmission of secret shares. This approach ensures that secret information can be successfully reconstructed when transmitted along a rainbow path. As a result, the adherence to the rainbow path condition during the transfer of the code from one node to another becomes pivotal; any deviation from this condition prevents the reconstruction of the secret by any other participants involved. By employing the RACN logic framework within a secret-share scheme, the process of reconstructing the desired information or data is facilitated.\\
\tab On further expanding this concept, in situations where the retrieved information needs to be disseminated to all participants in the scheme, a communication strategy is devised, involving rounds of communication. This approach ensures the secure and efficient distribution of the retrieved information to all relevant participants in the multi-participant group.

\section{RACN in Secret-Sharing}
\begin{figure}[h!]
	\begin{center}
		\tab[-0.7cm]\includegraphics[width=1.05\textwidth]{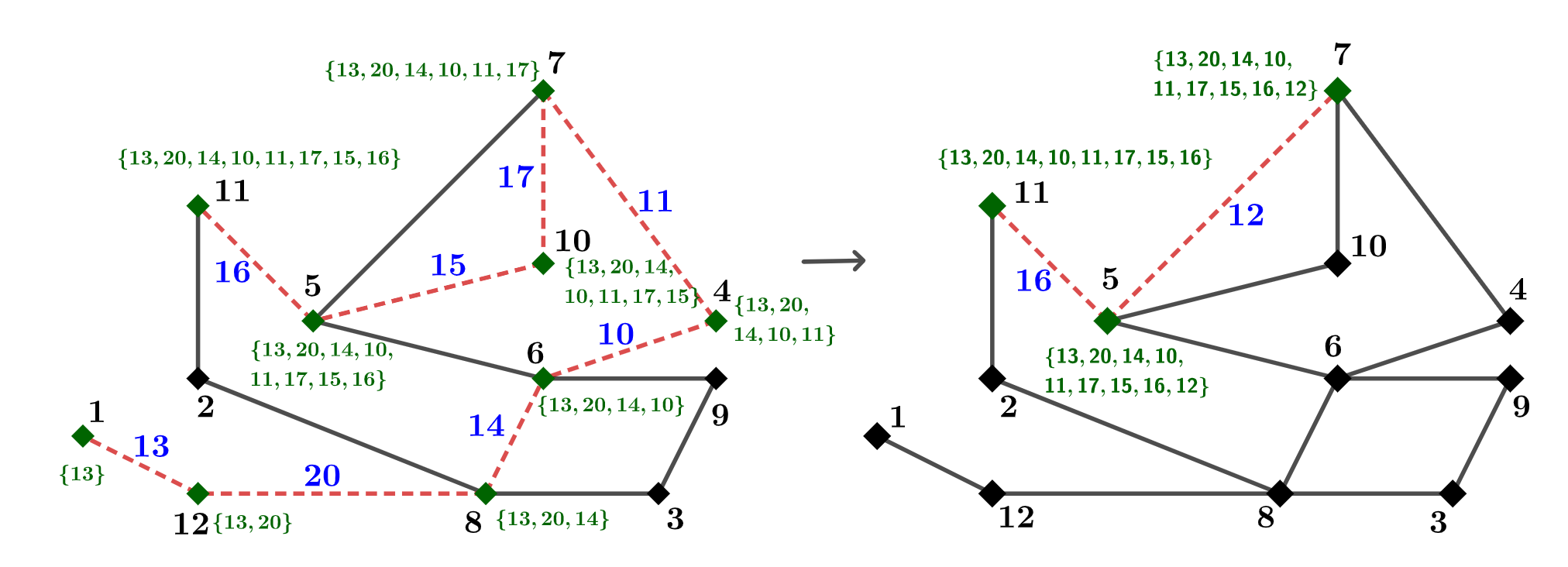}\\
		\text{a. Reconstruction Phase 1 \tab[4.5cm] b. Reconstruction Phase 2 }
		\caption{Illustration of Reconstruction phase using RACN}
	\end{center}
\end{figure}

\noindent \textbf{\textit{Initialization:}} This step involves choosing the secret that needs to be shared among $n$ participants. A $(k,m) - threshold$ is determined, where $k$ represents a threshold value of shares and $m$ represents a minimum number of participants required to reconstruct the secret.\\
\noindent \textbf{\textit{Share Distribution:}} Secret sharing allows for the distribution of critical information across multiple participants, reducing the risk associated with a single point of failure. This is especially important in scenarios where a single entity should not have full access to the secret. Each node representing a participant has access to a secret code that is represented by their respective edge weights. The distribution of secret codes among $n$ participants adheres to the principle of Rainbow Antimagic Connection Number (RACN) by splitting them into $k$ pieces of shares. The value of minimum $k$ signifies the least number of colors employed by the graph to fulfill the rainbow coloring prerequisite. As a result, a secret, represented by their respective edge weights, can only be accessed by a pair of participants. The approach of distributing these secrets can be adapted depending on the accessible communication channels and the desired level of security.\\
\noindent \textbf{\textit{Reconstruction Phase:}} This phase involves a series of sequential reconstruction iterations aimed at gathering all $k$ shares. Initially, a subset of secret codes is collected, and the remaining pieces are gathered through subsequent phases. For instance, from Figure 1., in the first reconstruction phase, a maximal rainbow path is established, encompassing a minimum of $n-k$ pieces, while ensuring coverage of up to $k-1$ pieces at most. Ideally, obtaining the longest rainbow path that covers all colors would be optimal if achieved in the initial iteration. When a minimum of $m$ participants come together they can collaboratively reconstruct the original secret using their shares.\\
\noindent\textbf{\textit{Privacy:}} Secret sharing ensures that individual participants only possess partial information and cannot deduce the complete secret on their own. This protects the privacy of the secret from any single participant. This scheme ensures that any subset of participants with fewer than $k$ shares cannot deduce any information about the original secret.\\
\noindent\textbf{\textit{Flexibility:}} The scheme may allow for dynamic addition or removal of participants without compromising the security of the shared secret. Though we consider only a static group of participants.\\
\tab Secret sharing schemes can have different variations, such as Shamir's Secret Sharing ~\cite{ref_Sha}, Blakley's Scheme~\cite{ref_Bla}, and more. The exact steps and mathematical techniques used can vary based on the chosen scheme, but the primary goal of the process is to ensure that the original secret remains confidential and can be reconstructed only when the required threshold of participants collaborates. Further innovations in this field are inspired from the following key references ~\cite{ref_Bei, ref_Bla, ref_BlBo, ref_Chi, ref_Fal, ref_Sha, ref_Cal_Bk}. Figure 1. illustrates the two stages of reconstruction, with parameters set at $n=12$, $k=9$, and $m=9$. In Figure 1(a), a depiction of the maximal longest rainbow path encompassing $k-1$ colors is presented. Notably, the longest rainbow path comprises only $8$ colors, necessitating the initiation of a successive phase. Figure 1(b) aptly depicts the culmination of the reconstruction phase, where participants $5$ and $7$ have successfully acquired all the individual share pieces.\\

\noindent\textit{\textbf{Problem 1. }How can the reconstructed secret be securely communicated to all participants in a multi-participant group after successful reconstruction?}\\
\tab Certainly, in the context of securely communicating the reconstructed secret to all participants within a multi-participant group, the process involves a series of well-defined steps to ensure the secret's confidentiality and integrity. Through multiple rounds of communication, the secret is effectively and securely conveyed to all participants. This is accomplished by employing cryptographic techniques, secure communication channels, and authentication mechanisms. Each step is carefully designed to prevent unauthorized access to the secret and safeguard its content from potential interception or tampering. The approach begins by encrypting the reconstructed secret using robust encryption algorithms. Encryption transforms the secret into an unreadable format, which can only be deciphered using the appropriate decryption keys possessed by authorized participants. This ensures that even if the communication is intercepted, the intercepted data remains unintelligible without the decryption keys.\\
\tab The secure communication of the reconstructed secret to all participants involves a strategic combination of encryption, secure communication channels, and authentication mechanisms. These measures collectively ensure that the secret's confidentiality remains intact, even in multiple rounds of communication, thereby upholding the security of the entire process.




\begin{figure}[h!]
	\begin{center}
		\includegraphics[width=1\textwidth]{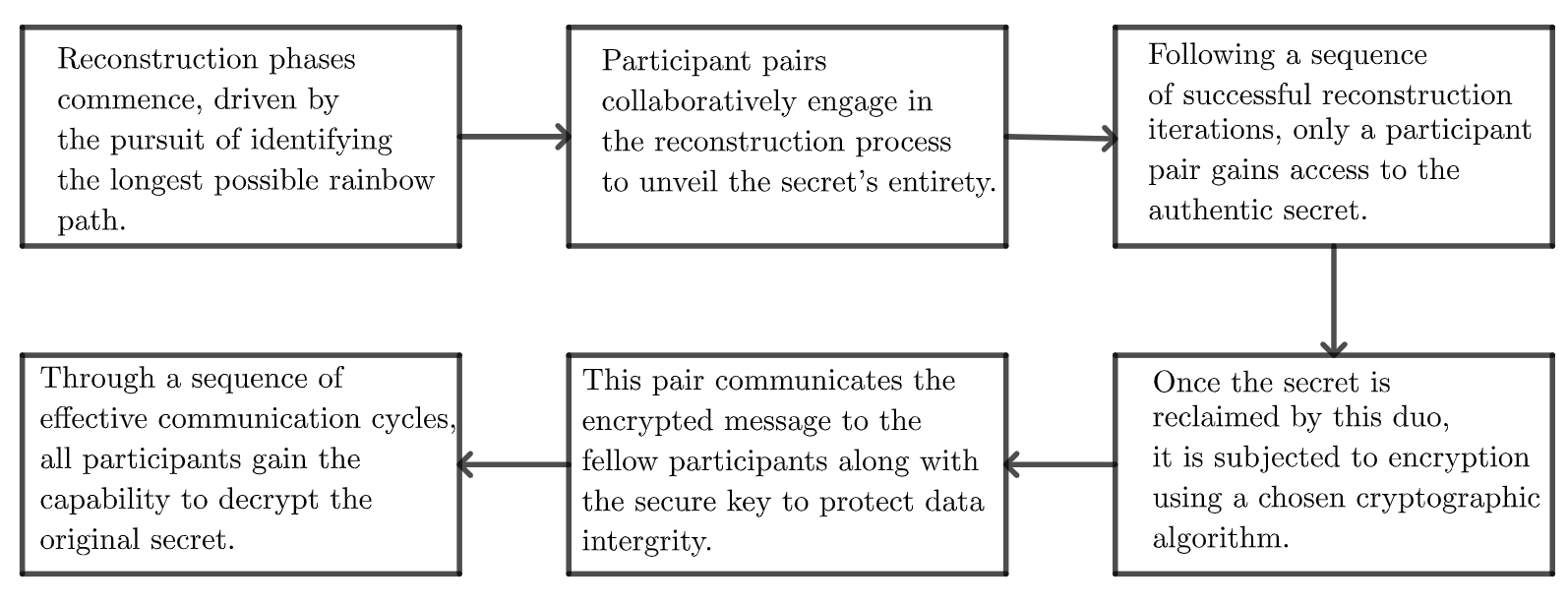}\\
		\caption{Illustration of Communication process}
	\end{center}
\end{figure}

\noindent \tab The transmission of the encrypted message can be effectively facilitated through closed circuits (i.e, cycles) within the network. For instance, from Figure 1., as the second reconstruction phase concludes, participants $5$ and $7$ successfully retrieve, reconstruct, and subsequently encrypt the secret for secure dissemination among the group. Therefore, during Round-1 communication, a closed circuit is established involving participants $5$, $7$, or both.
\begin{center}
$C_1 \Longleftrightarrow 7 - 10 - 5 - 6 - 4 - 7$\\
\tab[-1.6cm] $C_2 \Longleftrightarrow 7 - 10 - 5 - 7$\\
$C_3 \Longleftrightarrow 5 - 11 - 2 - 8 - 6 - 5$\\
\end{center}
\noindent However, $C_1$ and $C_3$ are identified as the most extensive and optimal paths for facilitating communication. Consequently, as Round-1 concludes, participants $2, 4, 5, 6, 7, 8, 10, 11$ have successfully received the encrypted secret.\\
\tab Moving to Round-2, examine all conceivable closed circuits encompassing both the aforementioned participants and those that remain unmentioned. Through this process, we obtain
\begin{center}
	$C_4 \Longleftrightarrow 6 - 9 - 3 - 8 - 6$ 
\end{center}
\noindent As Round-2 culminates, participants $1$ and $12$ remain the sole individuals who have not yet been engaged. However, there are no feasible closed circuits involving them. Consequently, the only participant capable of transmitting the encrypted secret becomes $8$. Hence, in Round-3, the communication proceeds as: 
$8 - 12 - 1$.\\

\noindent\textbf{\textit{Problem 2.}} \textit{How do factors like participant count, threshold values, and distribution method affect the minimum stages needed for successful secret reconstruction in this scheme?}\\
\tab The number of participants, threshold values, and distribution method play crucial roles in determining the minimum stages required for successful secret reconstruction in this scheme. The greater the number of participants, the potentially higher the number of stages needed for successful reconstruction. This is because each participant holds a share of the secret, and more participants could mean more stages to collaborate and reconstruct the original secret. The threshold value represents the minimum number of participants required to reconstruct the secret. If the threshold is lower, fewer participants are needed, possibly resulting in fewer stages. Conversely, a higher threshold requires more participants to be involved, which can lead to additional stages.\\
\tab Indeed, the method employed for distributing secret shares holds a significant influence on the minimum stages required for successful secret reconstruction. In our case, the utilization of the Rainbow Antimagic Connection Number (RACN) principle introduces a noteworthy advantage, which guarantees that a minimum stage is achieved when a longest rainbow path comprising all the edge weights is acquired in a single endeavor. By obtaining this maximal rainbow path in a single attempt, we can effectively minimize the number of stages needed for reconstruction. This aligns with the underlying philosophy of the RACN principle, enabling us to optimize the reconstruction process and expedite the successful recovery of the secret.\\
\tab These factors are interconnected. A higher participant count might demand a higher threshold, consequently leading to more stages. Additionally, the distribution method can influence both the efficiency of the process and the minimum stages required. Hence, a well-balanced consideration of these factors is essential to optimize the scheme's reconstruction process.\\
\tab The subsequent section presents the outcomes related to the minimum number of $k$-shares within an $n$-participant group. Noteworthy results from ~\cite{ref_Arp, ref_Aru, ref_Aug, ref_Bud, ref_Char, ref_Daf, ref_Nal, ref_Sep_Su, ref_Sep} are referred.

\section{RACN of some Graph Operations on Path graph}
\begin{note}
	The minimum number of shares $k$ varies with $p$, involving $m$ participants in an $n$-participant group, where $n \geq p$.
\end{note}


\begin{theorem}
	For $p \geq 2$, the RACN of shadow graph of path graph is,
	$$\displaystyle D^{(2)}_{rac}(P_p)=
	\begin{cases}
	p+1 & \text{for $p \equiv 0(mod\,2)$} \\
	p+3 & \text{for $p \equiv 1(mod\,2)$}
	\end{cases}
	$$
\end{theorem}
\begin{proof}
	Given that the vertex set of the shadow graph of a path graph is denoted as $\mathscr{V}(D^{(2)}(P_p)) = \{x_t,y_t : 1 \leq t \leq p\}$, it follows that the total number of participants in the $n$-participant group is $n=2p$.
	We know that, $\textbf{r}\mathfrak{ac}(P_p)=\mathscr{E}(P_p)=p-1$, hence we give the lower bound to be,\\
	$$\displaystyle D^{(2)}_{rac}(P_p) \geq
	\begin{cases}
	\mathscr{E}(P_p) + \delta(D^{(2)}(P_p)) & \text{for $p \equiv 0(mod\,2)$} \\
	\mathscr{E}(P_p) + \Delta(D^{(2)}(P_p)) & \text{for $p \equiv 1(mod\,2)$}
	\end{cases}
	$$
	
	\noindent The proof of the upper bound can be given by considering the following two cases.
	\begin{case}
		For $p \equiv 0(mod\,2)$
	\end{case}
	\tab Define an antimagic labeling $\varsigma_1: \mathscr{V}(D^{(2)}(P_p)) \longrightarrow \{1,2,3,\ldots,2p\}$ by: 
	$$
	\varsigma_1(x_t) =
	\begin{cases}
	2t-1 & \text{for $t=1,3,5,\ldots,p-1$}\\
	2t & \text{for $t=2,4,6,\ldots,p$}
	\end{cases}
	\tab[1.1cm]
	$$
	$$
	\varsigma_1(y_t) =
	\begin{cases}
	2p-2t+1 & \text{for $t=1,3,5,\ldots,p-1$}\\
	2p-2t+2 & \text{for $t=2,4,6,\ldots,p$}
	\end{cases}
	$$
	\noindent Based on the above labelling, the edge weights can be calculated as:
	\begin{alignat*}{4}
	W^1_{\varsigma_1}(x_tx_{t+1}) =\, & 4t+1 &\tab& \text{for $1 \leq t \leq p-1$}\\
	W^2_{\varsigma_1}(y_ty_{t+1}) =\, & 4p-4t+1 &\tab& \text{for $1 \leq t \leq p-1$}\\
	W^3_{\varsigma_1}(x_ty_{t+1}) =\, & 2p-1 &\tab& \text{for $1 \leq t \leq p-1$}\\
	W^4_{\varsigma_1}(y_tx_{t+1}) =\, & 2p+3 &\tab& \text{for $1 \leq t \leq p-1$}
	\end{alignat*}
	\noindent The total number of colors are obtained using N-term formula:
	$$
	\begin{array}{c c c c}
	U^{W^1_{\varsigma_1}}_N = a + (N-1)d \longleftrightarrow & 4p-3 = 5 + (N-1)(4) \longleftrightarrow & N = |W^1_{\varsigma_1}| = p-1\\
	U^{W^2_{\varsigma_1}}_N = a + (N-1)d \longleftrightarrow & 5 = 4p-3 + (N-1)(-4) \longleftrightarrow & N = |W^2_{\varsigma_1}| = p-1\\
	\end{array}
	$$
	\noindent Notably, $|W^3_{\varsigma_1}| = |W^4_{\varsigma_1}| = 1$. Hence, $\sum_{i=1}^{4}|W^i_{\varsigma_1}| = p - 1 + 2 = p+1$, i.e., $D^{(2)}_{rac}(P_p) \leq p+1$.
	
	\begin{case}
		For $p \equiv 1(mod\,2)$
	\end{case}
		\tab Define an antimagic labeling $\varsigma_2: \mathscr{V}(D^{(2)}(P_p)) \longrightarrow \{1,2,3,\ldots,2p\}$ by: 
		$$
		\varsigma_2(x_t) =
		\begin{cases}
		2t-1 & \text{for $t=1,3,5,\ldots,p$}\\
		2t & \text{for $t=2,4,6,\ldots,p-1$}
		\end{cases}
		\tab[1.1cm]
		$$
		$$
		\varsigma_2(y_t) =
		\begin{cases}
		2p-2t+2 & \text{for $t=1,3,5,\ldots,p$}\\
		2p-2t+1 & \text{for $t=2,4,6,\ldots,p-1$}
		\end{cases}
		$$
		\noindent Based on the above labelling, the edge weights can be calculated as:
		\begin{alignat*}{4}
			W^1_{\varsigma_2}(x_tx_{t+1}) =\, & 4t+1 &\tab& \text{for $1 \leq t \leq p-1$}\\
			W^2_{\varsigma_2}(y_ty_{t+1}) =\, & 4p-4t+1 &\tab& \text{for $1 \leq t \leq p-1$}\\
			W^3_{\varsigma_2}(x_ty_{t+1}) =\, & 2(p-1) &\tab& \text{for $t=1,3,5,\ldots,p-2$}\\
			W^4_{\varsigma_2}(x_ty_{t+1}) =\, & 2p &\tab& \text{for $t=2,4,6,\ldots,p-1$}\\
			W^5_{\varsigma_2}(y_tx_{t+1}) =\, & 2(p+2) &\tab& \text{for $t=1,3,5,\ldots,p-2$}\\
			W^6_{\varsigma_2}(y_tx_{t+1}) =\, & 2(p+1) &\tab& \text{for $t=2,4,6,\ldots,p-1$}
		\end{alignat*}
	\noindent The total number of colors are obtained using N-term formula:
	$$
	\begin{array}{c c c c}
	U^{W^1_{\varsigma_2}}_N = a + (N-1)d \longleftrightarrow & 4p-3 = 5 + (N-1)(4) \longleftrightarrow & N = |W^1_{\varsigma_2}| = p-1\\
	U^{W^2_{\varsigma_2}}_N = a + (N-1)d \longleftrightarrow & 5 = 4p-3 + (N-1)(-4) \longleftrightarrow & N = |W^2_{\varsigma_2}| = p-1\\
	\end{array}
	$$
	\noindent Notably, $|W^3_{\varsigma_2}| = |W^4_{\varsigma_2}| = |W^5_{\varsigma_2}| = |W^6_{\varsigma_2}| = 1$. Hence, $\sum_{i=1}^{6}|W^i_{\varsigma_2}| = p - 1 + 4 = p+3$, i.e., $D^{(2)}_{rac}(P_p) \leq p+3$.\\
	\noindent However, from Case 1 and Case 2,
	$$\displaystyle D^{(2)}_{rac}(P_p) \leq
	\begin{cases}
	p+1 & \text{for $p \equiv 0(mod\,2)$} \\
	p+3 & \text{for $p \equiv 1(mod\,2)$}
	\end{cases}
	$$
	\noindent On combining both upper and lower bounds, we get,
	$$\displaystyle D^{(2)}_{rac}(P_p) =
	\begin{cases}
	p+1 & \text{for $p \equiv 0(mod\,2)$} \\
	p+3 & \text{for $p \equiv 1(mod\,2)$}
	\end{cases}
	$$
\end{proof}
\begin{figure}[h!]
	\begin{center}
		\tab[-1cm]\includegraphics[width=0.55\textwidth]{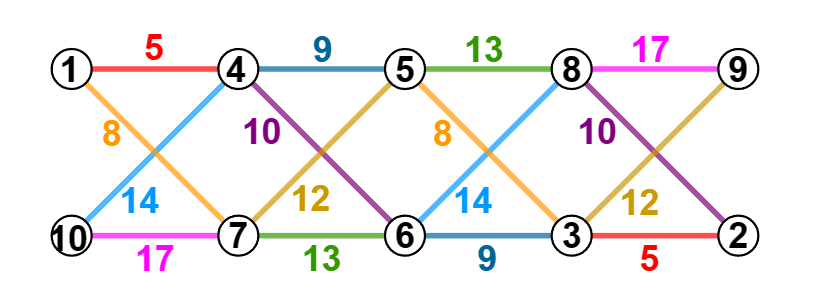}
		\caption{Illustration of RACN of shadow graph of Path graph}
	\end{center}
\end{figure}

\begin{theorem}
	For any $p \geq 2$, the RACN of splitting graph of path graph is,
	$$
	\displaystyle Spl^{(2)}_{rac}(P_p) = p+1 \tab \text{$\forall p \in N$}
	$$
\end{theorem}

\begin{proof}
	Given that the vertex set of the splitting graph of a path graph is denoted as $\mathscr{V}(Spl^{(2)}(P_p)) = \{x_t,y_t : 1 \leq t \leq p\}$, it follows that the total number of participants in the $n$-participant group is $n=2p$.
	We know that, $\textbf{r}\mathfrak{ac}(P_p)=\mathscr{E}(P_p)=p-1$, hence we give the lower bound to be,
	$$
	\displaystyle Spl^{(2)}_{rac}(P_p) \geq \mathscr{E}(P_p)+\delta(Spl^{(2)}(P_p))+1
	$$
	\noindent Inorder to prove the upper bound consider defining an antimagic labeling $\varsigma_3: \mathscr{V}(Spl^{(2)}(P_p)) \longrightarrow \{1,2,3,\ldots,2p\}$ by: 
	\begin{alignat*}{4}
		\varsigma_3(x_t) =\,& t &\tab& \text{for $1 \leq t \leq p$}\\
		\varsigma_3(y_t) =\,& 2p-t+1 &\tab& \text{for $1 \leq t \leq p$}
	\end{alignat*}
	\noindent Based on the above labelling, the edge weights can be calculated as:
	\begin{alignat*}{4}
	W^1_{\varsigma_3}(y_ty_{t+1}) =\, & 4p-2t+1 &\tab& \text{for $1 \leq t \leq p-1$}\\
	W^2_{\varsigma_3}(x_ty_{t+1}) =\, & 2p &\tab& \text{for $1 \leq t \leq p-1$}\\
	W^3_{\varsigma_3}(y_tx_{t+1}) =\, & 2(p+1) &\tab& \text{for $1 \leq t \leq p-1$}
	\end{alignat*}
	\noindent The total number of colors are obtained using N-term formula:
	$$
	\begin{array}{c c c c}
	U^{W^1_{\varsigma_3}}_N = a + (N-1)d \longleftrightarrow & 2p+3 = 4p-1 + (N-1)(-2) \longleftrightarrow & N = |W^1_{\varsigma_3}| = p-1\\
	\end{array}
	$$
	\noindent Notably, $|W^2_{\varsigma_3}| = |W^3_{\varsigma_3}| = 1$. Hence, $\sum_{i=1}^{3}|W^i_{\varsigma_3}| = p - 1 + 2 = p+1$, i.e., $Spl^{(2)}_{rac}(P_p) \leq p+1$.\\
	\noindent On combining the upper and lower bounds, we get,  
	$$
	Spl^{(2)}_{rac}(P_p) = p+1 \tab \text{$\forall p \in N$}
	$$
\end{proof}

\begin{figure}[h!]
	\begin{center}
		\tab[-1cm]\includegraphics[width=0.55\textwidth]{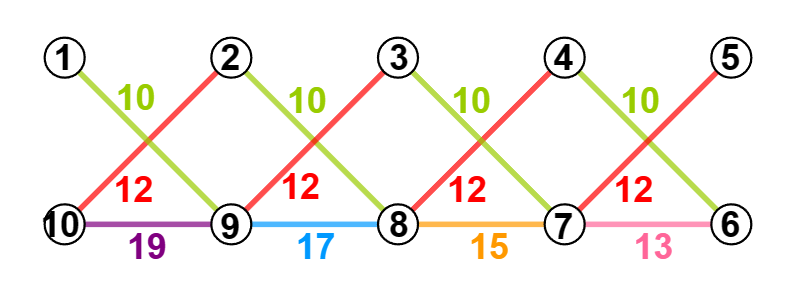}
		\caption{Illustration of RACN of splitting graph of Path graph}
	\end{center}
\end{figure}

\begin{theorem}
	For any $p \geq 2$, the RACN of mycielski graph of path graph is,
	$$
	\mu_{rac}(P_p) = 2p \tab \text{$\forall p \in N$}
	$$
\end{theorem}

\begin{proof}
	Given that the vertex set of the mycielski graph of a path graph is denoted as $\mathscr{V}(\mu(P_p)) = \{a\} \cup \{x_t,y_t : 1 \leq t \leq p\}$, it follows that the total number of participants in the $n$-participant group is $n=2p+1$.
	We know that, $\textbf{r}\mathfrak{ac}(P_p)=\mathscr{E}(P_p)=p-1$, hence we give the lower bound to be,
	$$
	\displaystyle \mu_{rac}(P_p) \geq \mathscr{E}(P_p)+\Delta(\mu(P_p))+1
	$$
	\noindent Inorder to prove the upper bound consider defining an antimagic labeling $\varsigma_4: \mathscr{V}(\mu(P_p)) \longrightarrow \{1,2,3,\ldots,2p+1\}$ by: 
	\begin{alignat*}{4}
	\varsigma_4(a) =\,& p+1 &\tab& \text{$\forall p \in N$}\\
	\varsigma_4(x_t) =\,& 2p-t+2 &\tab& \text{for $1 \leq t \leq p$}\\
	\varsigma_4(y_t) =\,& t &\tab& \text{for $1 \leq t \leq p$}\\
	\end{alignat*}
	\noindent Based on the above labelling, the edge weights can be calculated as:
	\begin{alignat*}{4}
	W^1_{\varsigma_4}(x_tx_{t+1}) =\, & 4p-2t+3 &\tab& \text{for $1 \leq t \leq p-1$}\\
	W^2_{\varsigma_4}(x_ty_{t+1}) =\, & 2p+3 &\tab& \text{for $1 \leq t \leq p-1$}\\
	W^3_{\varsigma_4}(y_tx_{t+1}) =\, & 2p+1 &\tab& \text{for $1 \leq t \leq p-1$}\\
	W^4_{\varsigma_4}(ay_t) =\, & p+t+1 &\tab& \text{for $1 \leq t \leq p$}
	\end{alignat*}
	
	\noindent As the edge weight $W^3_{\varsigma_4}$ has already been included in $W^4_{\varsigma_4}$, it is deemed superfluous and, therefore, it is excluded. The total number of colors are obtained using N-term formula:
	$$
	\begin{array}{c c c c}
	U^{W^1_{\varsigma_4}}_N = a + (N-1)d \longleftrightarrow & 2p+5 = 4p+1 + (N-1)(-2) \longleftrightarrow & N = |W^1_{\varsigma_4}| = p-1\\
	U^{W^4_{\varsigma_4}}_N = a + (N-1)d \longleftrightarrow & 2p+1 = p+2 + (N-1)(1) \tab[0.5cm]\longleftrightarrow & \tab[-0.8cm] N = |W^4_{\varsigma_4}| = p
	\end{array}
	$$
	\noindent Notably, $|W^2_{\varsigma_4}| = 1$. Hence, $\sum_{i=1,2,4}|W^i_{\varsigma_4}| = p - 1 + p + 1 = 2p$, i.e., $\mu_{rac}(P_p) \leq 2p$.\\
	\noindent On combining the upper and lower bounds, we get,  
	$$
	\mu_{rac}(P_p) = 2p \tab \text{$\forall p \in N$}
	$$
\end{proof}
\begin{figure}[h!]
	\begin{center}
		\tab[-1cm]\includegraphics[width=0.55\textwidth]{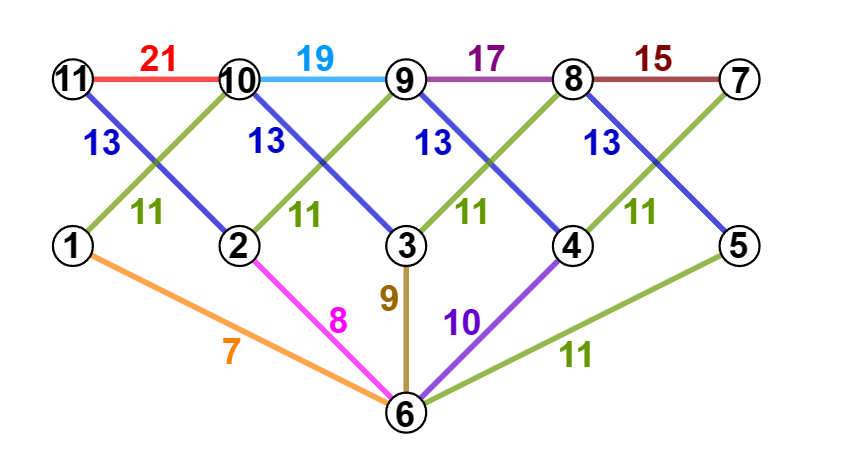}
		\caption{Illustration of RACN of mycielski graph of Path graph}
	\end{center}
\end{figure}

\noindent The above results yield the following key observations:
\begin{itemize}
	\item[1.] For a graph with $n=2p$ participants, the number of shares to be distributed are:
	$$a.\,\, \displaystyle D^{(2)}_{rac}(P_p)= k =
	\begin{cases}
	p+1 & \text{for $p \equiv 0(mod\,2)$} \\
	p+3 & \text{for $p \equiv 1(mod\,2)$}
	\end{cases}
	\tab[5.5cm]
	$$
	\tab[0.9cm]$b.\,\, \displaystyle Spl^{(2)}_{rac}(P_p)= k = p+1$
	
	\item[2.] For a graph with $n=2p+1$ participants, the number of shares to be distributed is:\\
	\tab[0.9cm]$a.\,\, \displaystyle \mu_{rac}(P_p)= k = 2p$\\
	
	\item[3.] The minimum number of participants $m$ required from an $n$-participant group is:\\
	\tab[0.9cm] $\displaystyle a. \,\, m(D^{(2)}_{rac}(P_p)) = \frac{1}{2}(2p+5+(-1)^{p-1})$\\
	$$\displaystyle b.\,\, m(Spl^{(2)}_{rac}(P_p)) =
	\begin{cases}
		p+2 & \text{for $p \geq 2 ; p \neq 3$}\\
		p+1 & \text{for $p = 3$}
	\end{cases}
	\tab[5.5cm]
	$$
	\tab[0.8cm] $\displaystyle c. \,\, m(\mu_{rac}(P_p)) = 2p+1$

	\item[4.] The number of reconstruction phases is:
	$$\displaystyle a.\,\, RP(D^{(2)}_{rac}(P_p)) =
	\begin{cases}
	1 & \text{when p is even}\\
	2 & \text{when p is odd}
	\end{cases}
	\tab[5.8cm]
	$$
	$$\displaystyle b.\,\, RP(Spl^{(2)}_{rac}(P_p)) =
	\begin{cases}
	1 & \text{for $p \geq 2 ; p \neq 3$}\\
	2 & \text{for $p = 3$}
	\end{cases}
	\tab[5.4cm]
	$$
	\tab[1.1cm] $\displaystyle c. \,\, RP(\mu_{rac}(P_p)) = \frac{1}{4}(2p+1+(-1)^{p-1})$\\
\end{itemize}

\section{Conclusion}
\tab The comprehensive analysis in the paper has investigated the rainbow antimagic connection number (RACN) of path graphs, utilizing some graph operations. It is pivotal to emphasize that ascertaining the precise value of RACN is a well-established NP-Hard problem, signifying its inherent complexity. Therefore, we propose subsequent open problems to stimulate further exploration in this area of research.
\begin{itemize}
	\item How does the scheme address the loss or corruption of a secret code? What measures are in place to ensure successful recovery of the original secret despite such occurrences?
	\item How does the absence of closed circuits impact the scheme and secret communication?
	\item How can the number of communication rounds needed post-secret recovery be minimized?
\end{itemize}

%
%

\end{document}